\documentclass[showpacs,amsmath,amssymb,10pt,aps]{revtex4}
\usepackage{amsfonts}
\newcommand{\ot}{{\,\otimes\,}}
\def\oper{{\mathchoice{\rm 1\mskip-4mu l}{\rm 1\mskip-4mu l}
{\rm 1\mskip-4.5mu l}{\rm 1\mskip-5mu l}}}
\def\<{\langle}
\def\>{\rangle}

\newtheorem{theorem}{Theorem}

\newtheorem{corollary}{Corollary}
\newtheorem{definition}{Definition}
\newtheorem{example}{Example}

\newtheorem{proposition}{Proposition}
\newtheorem{remark}{Remark}
\newenvironment{proof}[1][Proof]{\noindent\textbf{#1.} }{\ \rule{0.5em}{0.5em}}

\begin{document}

\title{\bf On Bipartite Operators Defined by Completely Different Permutations }
\author{ Marek Mozrzymas$^1$, Dariusz Chru{\'s}ci{\'n}ski$^2$ and Gniewomir Sarbicki$^2$ }
\affiliation{ $^1$Institute for Theoretical Physics, University of Wroc{\l}aw, 50-204 Wroc{\l}aw, Poland \\
$^2$Institute of Physics, Faculty of Physics, Astronomy and Informatics \\  Nicolaus Copernicus University,
Grudzi\c{a}dzka 5/7, 87--100 Toru\'n, Poland}

\begin{abstract}
We introduce a class of bipartite operators acting on $\mathcal{H} \ot  \mathcal{H}$ ($\mathcal{H}$ being an $n$-dimensional Hilbert space) defined by a set of $n$ Completely Different Permutations CDP. Bipartite operators are of particular importance in quantum information theory to represent states and observables of composite quantum systems.  It turns out that any set of CDPs gives rise to a certain direct sum decomposition of the total Hilbert space which enables one simple construction of the corresponding bipartite operator. Interestingly, if set of CDPs defines an abelian group then the corresponding bipartite operator displays an additional property -- the partially transposed operator again corresponds to (in general different) set of CDPs. Therefore, our technique may be used to construct new classes of so called PPT states which are of great importance for quantum information. Using well known relation between bipartite operators and linear maps one use also construct linear maps related to CDPs.
\end{abstract}

\maketitle

\section{Introduction}

Quantum entanglement is a basic physical resource for modern quantum technologies like quantum teleportation, quantum computation, quantum communication and quantum cryptography \cite{QIT,HHHH}. It is therefore clear that detailed analysis of the mathematical structure of quantum states represented by bipartite density operators is of great importance for quantum information theory. However, in general it is very hard to check whether
a given density matrix describing a quantum state of the composite system is separable or entangled (so called separability problem). It was shown by Gurvits that the separability  problem is NP hard \cite{NP}.

There are several operational criteria which provide necessary conditions for separability and sufficient conditions for entanglement \cite{Guhne,HHHH,TOPICAL}.  The most simple is the celebrated  Peres-Horodecki criterion \cite{HHHH} based on the operation of
partial transposition: if a state $\rho_{AB}$  is separable then its partial
transposition $(\oper \ot T)\rho$  is positive. States which are positive
under partial transposition are called PPT states. Clearly each separable state is necessarily PPT but the converse is not true. These two sets coincide only for $2 \ot 2$ and $2 \ot 3$ systems \cite{HHHH}.

The paper is organised as follows. In Section II we introduce the concept of
Completely Different Permutations CDP and we derive the basic properties
of maximal sets of Completely Different Permutations. It appears that these
stets have many ineteresing and usefull properties, in particular it is
shown that any commutative set of CDP is necessarily an abelian group. In
next Section we consider groups of CDP and shown in particular that
wellknown Cayley construction of permutational representations of finite
groups leads to the groups of CDP. In Section IV we define a class of
tensor product matrices, wich we call CDP matrices, whose construction is
based on the properties of sets of CDPs. The construction is a generalisation of
construction given in paper \cite{circ}, where the cyclic group of order n,
a particular case of CDP, where used. Using the properties of sets of CDPs we
derive the basic properties tensor product matrices. In particular we derive
the direct sum decomposition of the carrier space of CDP matrices and the
corresponding decomposition of CDP matrices into a direct sum of
orthogonally supported operators. Further we derive some properties of
partially transposed CDP matrices. It appears that, under assumptions of
commutativity of the set of CDPs definig the CDP matrix, the partially
transposed CDP matrix is again a CDP matrix with, in general,  a
transformed set of CDP. This allows formulate a conditions for PPT
property of the CDP matrix. In next subsections we give a realignement and
majorisation criteria for tensor product matrices. In the last \ V section
we consider the properties of some linear maps: particular Irreducibly
Covariant Quantum Channels, the Reduction map and of its generalisation the
Breuer-Hall map \cite{B,H}  and we describe their relation, via Choi-Jamiolkowski
isomorphism,  to CDP\ matrices.

\section{Completely Different Permutations}

In order to generalize the idea of circulant states introduced in \cite{circ} we introduce a concept  of Completely Different
Permutations CDP. In what follows we denote by $S(n)$ the symmetric group $S(n)$ and $\sigma \in S(n)$ denotes a permutation with the corresponding matrix representation $m(\sigma)=(\delta _{\sigma^{-1}(i)j})$. A cyclic group generated by a cycle
permutation $c$ of length $n$ will be denoted $C(n)= \{c^{i}=(c)^{i}:i=0,1,\ldots,n-1\}\subset S(n)$ where $c^{0}= {\rm id},
c^{i}(j)=j+_{n}i \equiv i+j\quad \text{mod}(n).$

\begin{definition} \label{CDPdef} Two permutations $\sigma, \rho \in S(n)$ are Completely Different CDP iff $\sigma(i)\neq \rho(i)$ for any $i=1,\ldots,n$.
\end{definition}

It easy to check that

\begin{proposition} \label{CDP-MOP}
Permutations $\sigma,\rho \in S(n)$ are CDP  iff $tr(m(\sigma
^{-1})m(\rho ))=0$.
\end{proposition}
Hence $\sigma ,\rho \in S(n)$ are CDP iff $m(\sigma)$ and $m(\rho)$ are mutually orthogonal  with respect to the Frobenius
scalar product in $M(n,\mathbb{C})$. Therefore, one may equivalently call a set of CDP a set of  Mutually Orthogonal Permutations (MOP).


One can check that

\begin{proposition} \label{at_most_n}
Any set of CDP in $S(n)$ contain at most $n$ elements.
\end{proposition}
In what follows we  consider only maximal sets of CDPs containing $n$
permutations and they will be denoted $\Sigma _{n}=\{\sigma
_{i}\}_{i=1}^{n}$. The sets of CDPs have many interesting properties and the first one is the
possibility to enumerate the permutations from CDP $\Sigma _{n}$ in a very
convenient and useful way. The structure of a set of CDPs allows to enumerate a set $\{\sigma_1,\ldots,\sigma_n\}$ as follows

\begin{equation} 
\sigma _{i}(1)=i,\qquad i=1,...,n \ .
\end{equation}
Using this convention one finds

\begin{proposition} \label{abelianCDP}
Let $\Sigma _{n}$ $=\{\sigma _{i}\}_{i=1}^{n}$ be an abelian set of CDP.  Then

\begin{equation} 
\forall i,j\qquad \sigma _{i}(j)=\sigma _{j}(i),
\end{equation}
and
\begin{equation} 
\forall k\quad \forall i,j\qquad \sigma _{\sigma _{k}(i)}(j)=\sigma _{\sigma
_{k}(j)}(i).
\end{equation}
\end{proposition}

\begin{proof}
If $\Sigma _{n}$ $=\{\sigma _{i}\}_{i=1}^{n}$ is an abelian CDP, then
\begin{equation} 
\forall k\quad \forall i,j\qquad \sigma _{i}\sigma _{j}(k)=\sigma _{j}\sigma
_{i}(k),
\end{equation}%
in particular it holds for $k=1$ which in our enumeration $\sigma
_{i}(1)=i $ gives
\begin{equation} 
\forall i,j\qquad \sigma _{i}\sigma _{j}(1)=\sigma _{j}\sigma
_{i}(1)\Rightarrow \sigma _{i}(j)=\sigma _{j}(i).
\end{equation}%
The second statement  simply follows.
\end{proof}

Let us consider some examples of sets of CDPs. The simplest one is the group $S(2) = \{ \mathrm{Id}, (12) \}$
which is obviously a commutative set of CDPs. The next example is also simple but
less trivial

\begin{example} \label{ex_S3}
In the group $S(3)$ we have two sets of CDPs
\begin{equation} 
\Sigma _{3}=\{\sigma _{1}=(23),\sigma _{2}=(12),\sigma _{3}=(13)\},\qquad
C(3)=\{c^{i}=(012)^{i}:i=0,1,2\}.
\end{equation}
\end{example}

\begin{example} \label{ex_circ}
Any cyclic group generated by a cycle permutation $c=(01..n-1)$ of length $n$
\begin{equation} 
C(n)=\{c^{i}=(c)^{i}:i=0,1,...,n-1\}
\end{equation}
is a set of CDPs.
\end{example}


\begin{example} \label{ex_V4}
The abelian group
\begin{equation} 
V(4)=\{\sigma _{1}={\rm id},\ \sigma _{2}=(12)(34),\ \sigma
_{3}=(13)(24),\ \sigma _{4}=(14)(23)\}\subset S(4)
\end{equation}
defines a set of CDPs. However, the following sets of CDPs
\begin{equation} 
\Sigma _{4}=\{\sigma _{1}=(34),\quad \sigma _{2}=(12),\quad \sigma
_{3}=(13)(24),\quad \sigma _{4}=(14)(23)\}\subset S(4),
\end{equation}%
and
\begin{equation} 
\Sigma _{4}^{\prime }=\{\sigma _{1}=(234),\quad \sigma _{2}=(124),\quad
\sigma _{3}=(132),\quad \sigma _{4}=(143)\}\subset S(4)
\end{equation}
are not groups.
\end{example}

\begin{example} \label{ex_S5}
The following subset of $S(5)$
\begin{equation} 
\Sigma _{5}=\{\sigma _{1}= {\rm id},\;\sigma _{2}=(12)(345),\;\sigma
_{3}=(13)(542),\;\sigma _{4}=(14)(352),\;\sigma _{5}=(15)(243)\}
\end{equation}%
defines CDP.
\end{example}

\begin{proposition} \label{relacja}
Let $\Sigma _{n}=\{\sigma _{i}\}_{i=1}^{n}\subset S(n)$ be a set of CDPs. Then
the sets $\Sigma _{n}^{-1}=\{\sigma _{i}^{-1}\}_{i=1}^{n}$ and $\rho \Sigma
_{n}\delta ,$ for any $\rho ,\delta \in S(n)$ are also a set of CDPs.
\end{proposition}


One can check that in the Example \ref{ex_S3} we have in $S(3)$%
\begin{equation} 
\Sigma _{3}=C(3)(23)
\end{equation}%
where we have changed the notation for the cycle $c=(123).$ In Example 12
for $S(4)$ we have%
\begin{equation} 
\Sigma _{4}=(34)C(4)
\end{equation}%
where $C(4)=\{\sigma _{1}=id,$ \ $\sigma _{2}=(12)(34),$ $\ \sigma
_{3}=(1324),\quad \sigma _{4}=(1423)\}$ and
\begin{equation} 
(23)\Sigma _{4}^{\prime }(34)=V(4).
\end{equation}%



\begin{proposition} \label{fixed_points}
Let $\Sigma _{n}=\{\sigma _{i}\}_{i=1}^{n}\subset S(n)$ be a set of CDP. Then

\begin{equation} 
\forall j=1,..,n\quad \exists i=1,..,n\qquad \sigma _{i}(j)=j,
\end{equation}
i.e. any $j=1,..,n$ is a fixed point for some permutation in the set of CDPs $\Sigma
_{n}$. Moreover,

\begin{equation} 
\forall i=1,..,n\quad \exists j=1,..,n\qquad \sigma _{i}(j)=1.
\end{equation}
\end{proposition}
Due to our enumeration convention we have always $\sigma _{1}(1)=1.$ If
the identity permutation $id=$ $\sigma _{1}$ is in a set of CDPs $\Sigma _{n}$ then
obviously all $j=1,..,n$ are fixed points of $\sigma _{1}$ and all remaining $
\sigma _{i}\in \Sigma _{n}$ have no fixed points.


\begin{definition} \label{mE}
For any set of CDPs $ \Sigma _{n}=\{\sigma _{i}\}_{i=1}^{n}\subset S(n)$ we
define a set of matrices $m(E)=\{m(E_{j}),\quad j=1,..,n\}\quad \quad $
\begin{equation} 
\forall j=1,..,n\quad m(E_{j})=\sum_{k=1}^{n}e_{\sigma _{k}(1)\sigma
_{k}(j)}=\sum_{k=1}^{n}e_{k\sigma _{k}(j)}\in M(n,%
\mathbb{C}
)
\end{equation}%
where $\{e_{ij}\}_{i,j=1}^{n}$ is a natural basis of the linear space $M(n,%
\mathbb{C}
)$. In particular we have $m(E_{1})=\mathbf{1}_{n}\in M(n,%
\mathbb{C}
).$
\end{definition}

It is not difficult to check that the structure of the set of CDPs $ \Sigma _{n}$ implies

\begin{proposition} \label{CDP_from_mE}
Let $\Sigma _{n}=\{\sigma _{i}\}_{i=1}^{n}\subset S(n)$ be a set of CDPs then for
any $j=1,..,n$ the matrices $m(E_{j})=\sum_{k=1}^{n}e_{k\sigma _{k}(j)}$ are
mutually orthogonal e.i.
\begin{equation} 
tr(m(E_{i})^{+}m(E_{j}))=\delta _{ij}n.
\end{equation}%
The matrices $m(E_{j})$ are permutation matrices, which are natural matrix
representations of corresponding permutations denoted $E_{j}\in S(n)$ which
are of the form
\begin{equation} 
E_{j}=\left(
\begin{array}{ccccc}
\sigma _{1}(j) & \sigma _{2}(j) & . & . & \sigma _{n}(j) \\
1 & 2 & . & . & n%
\end{array}%
\right) ,\qquad E_{j}(\sigma _{k}(j))=k,
\end{equation}%
in particular $E_{1}=id$ and $E_{j}(\sigma _{1}(j))=1,$ so if $\sigma
_{1}=id\in \Sigma _{n}$ then $E_{j}(j)=1.$ Moreover the set of permutations $%
E\equiv \{E_{j}:j=1,..,n\}$ is also a set of CDPs, which is a direct consequence of
the mutual orthogonality of the matrices $m(E_{j})$.
\end{proposition}

\begin{definition} \label{conjugated}
The CDP $E\equiv \{E_{j}:j=1,..,n\}$ will be called the set of 
permutations conjugated to the set of CDPs $\ \Sigma _{n}=\{\sigma _{i}\}_{i=1}^{n},$ which
generates the matrices $m(E_{j})$.
\end{definition}

The set of matrices $m(E)\equiv \{m(E_{j}):j=1,..,n\}$, defined in
Def. \ref{mE} , which is a matrix representation of the set $E$ has the following
nice property

\begin{proposition} \label{abelian_to_group_mE}
Let $\Sigma _{n}=\{\sigma _{i}\}_{i=1}^{n}\subset S(n)$ is an abelian set of CDPs. Then the set of matrices $m(E)\equiv \{m(E):j=1,..,n\}$ is an abelian group
with the following composition law%
\begin{equation} 
\forall i,j=1,..,n\quad m(E_{i})m(E_{j})=m(E_{\sigma _{i}(j)}),
\end{equation}%
which follows directly from the definition of $m(E_{i})$. From Def \ref{mE} it
follows that $m(E_{1})=\mathbf{1}_{n}\in M(n,%
\mathbb{C}
)$.
\end{proposition}

Now because the natural matrix representation of $S(n)$ is
faithful, we get

\begin{corollary} \label{abelian_to_group}
If $\Sigma _{n}=\{\sigma _{i}\}_{i=1}^{n}\subset S(n)$ is an abelian set of CDPs,
then the set of CDPs $E\equiv \{E_{j}:j=1,..,n\}$ is an abelian group with
composition law of the form
\begin{equation} 
\forall i,j=1,..,n\quad E_{i}E_{j}=E_{\sigma _{i}(j)}.
\end{equation}
\end{corollary}

On the other hand the set $E$ is related to the abelian set of CDPs $\Sigma
_{n}=\{\sigma _{i}\}_{i=1}^{n}$ in the following way

\begin{proposition} \label{ab->conj=inv}
Suppose that $\Sigma _{n}=\{\sigma _{i}\}_{i=1}^{n}\subset S(n)$ is an
abelian set of CDPs, then the permutations in the set $E\equiv \{E_{j}:j=1,..,n\}$
are related to the permutations of the set $\Sigma _{n}$ in a very simple
way
\begin{equation} 
\forall i=1,..,n\quad E_{i}=\sigma _{i}^{-1}\Rightarrow E=\Sigma _{n}^{-1},
\end{equation}%
which follows from the relation for its matrix representations%
\begin{equation} 
\forall i=1,..,n\quad m(E_{i})=m(\sigma _{i}^{-1})\Rightarrow M(E)=M(\Sigma
_{n}^{-1})
\end{equation}%
and this is a simple consequence of Prop. \ref{abelianCDP}
\end{proposition}

Finally, one arrives at the following

\begin{theorem} \label{abelian->group}
Let $\Sigma _{n}=\{\sigma _{i}\}_{i=1}^{n}\subset S(n)$ be a set of CDPs. Then the
set $\Sigma _{n}$ is abelian if and only if it is an abelian group of
permutations realised as a set of CDPs.
\end{theorem}

\begin{remark} \label{repr_not_CDP}
Note that a finite groups may be realised as a groups of
permutations, which are not CDPs. For example the group $V(4)$ in Example
\ref{ex_V4}  is isomorphic to the following group of permutations $G=\{id,$ \ $(12),$
\ $(34),$ \ $(12)(34)\}$, which are not CDP.
\end{remark}

The correspondence between the sets $\Sigma _{n}$ and $E\equiv
\{E_{j}:j=1,..,n\}$ is not unique, in fact we have

\begin{proposition} \label{centraliser}
Let $\Sigma _{n}=\{\sigma _{i}\}_{i=1}^{n}\subset S(n)$ be a set of CDPs with $%
m(E)\equiv \{m(E_{j}):j=1,..,n\}$, defined in Def. \ref{mE}. Suppose that $\rho \in
S(n):\rho (1)=1$ and $\delta \in S(n)$ is such that
\begin{equation} 
m(\delta )M(E)m(\delta )^{-1}=M(E),
\end{equation}%
i.e. $m(\delta )$ belongs to the centraliser of the set $m(E)$, then the set of CDPs $\Sigma _{n}^{\prime }=\{\sigma _{i}^{\prime }=\delta \sigma _{i}\rho
\}_{i=1}^{n}$ is such that
\begin{equation} 
m(E_{i}^{\prime })=m(\delta )m(E_{\rho (i)})m(\delta )^{-1}\Rightarrow
m(E^{\prime })=(E),
\end{equation}%
so the set of CDPs $\Sigma _{n}$ and $\Sigma _{n}^{\prime }$ generates the same
set $M(E)$ and consequently the same conjugated set of CDPs $E$.
\end{proposition}

Let us consider some examples.

\begin{example} \label{ex_S3_2_CDP}
In the group $S(3)$ we have two sets of CDPs
\begin{equation} 
\Sigma _{3}=\{\sigma _{1}=(23),\sigma _{2}=(12),\sigma _{3}=(13)\},
\end{equation}
and
\begin{equation} 
\Sigma _{3}^{\prime }=\{\sigma _{1}^{\prime }=id,\sigma _{2}^{\prime
}=(123),\sigma _{3}^{\prime }=(132)\}=\Sigma _{3}(23),
\end{equation}
such that
\begin{equation} 
M(E^{\prime })=M(E)\Leftrightarrow E^{\prime }=E=\Sigma _{3}^{\prime }.
\end{equation}
\end{example}

\begin{example} \label{Sigma_4->C4}
For $\Sigma _{4}=\{\sigma _{1}=(34),\quad \sigma _{2}=(12),\quad \sigma
_{3}=(13)(24),\quad \sigma _{4}=(14)(23)\}\subset S(4)$ we have
\begin{equation} 
E_{1}=id,\quad E_{2}=(12)(34),\quad E_{3}=(1324),\quad E_{4}=(1423),
\end{equation}%
so in this case we have
\begin{equation} 
E=\{E_{j}:j=1,..,4\}=C(4)=\{(1324)^{k}:k=1,..,4\},
\end{equation}%
so it is a cyclic group and the matrices $\{m(E_{j}):j=1,..,n\}$ are simply
natural matrix representation of this group.
\end{example}

\begin{example} \label{Sigma4->V4}
Let $\Sigma _{4}=\{\sigma _{1}=(234),\quad \sigma _{2}=(124),\quad \sigma
_{3}=(132),\quad \sigma _{4}=(143)\}\subset S(4),$ then
\begin{equation} 
E_{1}=id,\quad E_{2}=(13)(24),\quad E_{3}=(14)(23),\quad E_{4}=(12)(34),
\end{equation}%
and we have
\begin{equation} 
E=\{E_{j}:j=1,..,4\}=V(4),
\end{equation}%
so again it is an abelian group although the set $\Sigma _{4}$ is neither
group nor a commutative set.
\end{example}

\begin{example} \label{Sigma5->Sigma5}
The set
\begin{equation} 
\Sigma _{5}=\{\sigma _{1}=id,\;\sigma _{2}=(12)(345),\;\sigma
_{3}=(13)(542),\;\sigma _{4}=(14)(352),\;\sigma _{5}=(15)(243)\}\subset S(5)
\end{equation}%
is such that
\begin{equation} 
E=\{E_{j}:j=1,..,5\}=\Sigma _{5}.
\end{equation}
\end{example}

In the next section we will consider sets of CDPs which are groups.

%



\section{Groups of CDPs}

In previous section we have presented some examples of sets of CDPs which were
groups and we have proved a remarkable property of abelian sets of CDPs, which are
always groups. In general finite groups are very rich source of sets of CDPs,
which follows from well-known construction of permutation representations of finite groups.
In fact we have

\begin{proposition} \label{regular}
Let $G=\{g_{i}:i=1,...,n\}$ be a finite group. Then its regular
representation, wich is in fact a permutation representation
\begin{equation} 
R(g_{i})\equiv \sigma _{i}\equiv \left(
\begin{array}{ccccc}
g_{1} & g_{2} & . & . & g_{n} \\
g_{i}g_{1} & g_{i}g_{2} & . & . & g_{i}g_{n}%
\end{array}%
\right) \in S(n)
\end{equation}%
is such that the group of permutations $R(G)=\{R(g_{i}):i=1,...,n\}\subset
S(n)$ is a set of CDPs.
\end{proposition}

\begin{proof}
Suppose that for some $k=1,....,n$
\begin{equation} 
R(g_{i})(g_{k})=R(g_{j})(g_{k})\Leftrightarrow
g_{i}g_{k}=g_{j}g_{k}\Longleftrightarrow g_{i}=g_{j},
\end{equation}%
so the permutations $R(g_{i})=\sigma _{i}$ and $R(g_{j})=\sigma _{j}$ are
completely different for any $i\neq j=1,...,n.$
\end{proof}

The cyclic groups in Examples \ref{ex_circ}, \ref{ex_V4} were, in fact regular representations.
The smallest nonabelian group is $S(3),$ of rank $6$ and its regular
representation may be presented as a set of CDPs as follows

\begin{example} \label{R(S3)}
$R(S(3))=\{\sigma _{1}=id,\;\sigma _{2}=(123)(456),\;\;\sigma
_{3}=(132)(465),\;\sigma _{4}=(14)(26)(53),\;\sigma
_{5}=(15)(24)(36),\;\sigma _{6}=(16)(25)(34)\}.$
\end{example}

\begin{proposition} \label{GCDP}
If $\Sigma _{n}=\{\sigma _{i}\}_{i=1}^{n}\subset S(n)$ is a group of CDPs, then

\begin{enumerate}

\item the composition rule in the group $\Sigma _{n}$ has remarkable simple
form
\begin{equation} 
\sigma _{i}\sigma _{j}=\sigma _{\sigma _{i}(j)}
\end{equation}

i.e. the matrix $M=(\sigma _{i}(j))$ describe the table of composition of
the group $\Sigma _{n}=\{\sigma _{i}\}_{i=1}^{n},$

\item the inverse elements are the following%
\begin{equation} 
(\sigma _{i})^{-1}=\sigma _{j}\Longleftrightarrow \sigma _{i}(j)=1=\sigma
_{j}(i),
\end{equation}

\item $\sigma _{1}=id$ and all the remaining elements $\sigma _{i}\in \Sigma_{n}$ have no fixed points.

\end{enumerate}

\end{proposition}

\begin{proof}
a) For arbitrary $\sigma _{i},\sigma _{j}$ belonging to a set of CDPs $ \Sigma _{n}$ there exist $\sigma _{k}\in \Sigma _{n}$ such that
\begin{equation} 
\sigma _{i}\sigma _{j}=\sigma _{k}\Rightarrow \sigma _{i}\sigma
_{j}(1)=\sigma _{k}(1)\Rightarrow \sigma _{i}(j)=k,
\end{equation}%
because in our enumeration,  we have $\forall i=1,..,n$ $\sigma
_{i}(1)=i.$ The remaining statements of the Proposition follow easily from
the structure of a set CDPs.
\end{proof}

If $\Sigma _{n}=\{\sigma _{i}\}_{i=1}^{n}\subset S(n)$ is a group of CDPs, then its
group structure strongly determines the structure of the set of CDPs $E=\{E_{j}:j=1,..,n\}$.

\begin{proposition} \label{GCDPconj}
If $\Sigma _{n}=\{\sigma _{i}\}_{i=1}^{n}\subset S(n)$ is a group of CDPs, then

\begin{enumerate}

\item  the set of CDPs $E=\{E_{j}:j=1,..,n\}$ is a group isomorphic with the group $%
\Sigma _{n}.$ The isomorphism is given by the following map%
\begin{equation} 
f(\sigma _{i})=E_{i},\quad i=1,..,n
\end{equation}%
which satisfy the isomorphism relation $f(\sigma _{i})f(\sigma
_{j})=f(\sigma _{i}\sigma _{j})$ because, one can check that, we have
\begin{equation} 
m(E_{i})m(E_{j})=m(E_{\sigma _{i}(j)}),
\end{equation}%
so the composition rule for the set of CDPs $E=\{E_{j}:j=1,..,n\}$ is the same as
for the group $\Sigma _{n}=\{\sigma _{i}\}_{i=1}^{n}$. If the group $\Sigma
_{n}$ is abelian then $\Sigma _{n}=E.$

\item the matrix groups $M(\Sigma _{n})=\{m(\sigma _{i})\}_{i=1}^{n}$ and $%
M(E)=\{m(E_{j}):j=1,..,n\}$ are mutually commutant e.i. we have
\begin{equation} 
\forall i,j=1,..,n\qquad m(\sigma _{i})m(E_{j})=m(E_{j})m(\sigma _{i}).
\end{equation}%
note that these groups mutually commute even if the set of CDPs $\Sigma _{n}$ is
not commutative.
\end{enumerate}

\end{proposition}

\begin{proof}
2. We have $\forall i,j=1,..,n$
\begin{equation} 
m(\sigma _{i})m(E_{j})=\sum_{k=1}^{n}e_{k\sigma
_{i}^{-1}(k)}\sum_{l=1}^{n}e_{l\sigma _{l}(j)}=\sum_{k,l=1}^{n}\delta
_{\sigma _{i}^{-1}(k)l}e_{k\sigma _{l}(j)}=\sum_{k=1}^{n}e_{k\sigma _{\sigma
_{i}^{-1}(k)}(j)}
\end{equation}%
and from composition rule in the group $\Sigma _{n}=\{\sigma
_{i}\}_{i=1}^{n}$ we get
\begin{equation} 
m(\sigma _{i})m(E_{j})=\sum_{k=1}^{n}e_{k\sigma _{i}^{-1}\sigma _{k}(j)}.
\end{equation}%
On the other hand we have
\begin{equation} 
m(E_{j})m(\sigma _{i})=\sum_{k=1}^{n}e_{k\sigma
_{k}(j)}\sum_{l=1}^{n}e_{l\sigma _{i}^{-1}(l)}=\sum_{k,l=1}^{n}\delta
_{\sigma _{k}(j)l}e_{k\sigma _{i}^{-1}(l)}=\sum_{k=1}^{n}e_{k\sigma
_{i}^{-1}\sigma _{k}(j).}
\end{equation}
\end{proof}

Let us check this on some examples. Let us check this on some examples.

\begin{example} \label{Sigma6}
If

$\Sigma _{6}=S(3)=\{\sigma _{1}=id,\;\sigma _{2}=(123)(456),\;\;\sigma
_{3}=(132)(465),\;\sigma _{4}=(14)(26)(53),\;\sigma
_{5}=(15)(24)(36),\;\sigma _{6}=(16)(25)(34)\},$

then

$E=\{E_{1}=id,\;E_{2}=(132)(456),\;\;E_{3}=(123)(465),\;E_{1}=(14)(25)(63),%
\;E_{5}=(15)(26)(34),\;E_{6}=(16)(24)(35)\},$

so it is again the group $S(3)$ but differently embedded in $S(6).$
\end{example}

\begin{example} \label{cycl_conj}
Let $\Sigma _{n}=C(n)=\{c^{i}=(c)^{i}:i=0,1,...,n-1\}\subset $ $S(n)$ where $%
c=(012...n-1)\Rightarrow c^{k}(t)=t+_{n}k.$ Then one can check that
\begin{equation} 
E_{i}=(c^{i})^{-1},
\end{equation}%
and therefore in this case we have $E=\Sigma _{n}=C(n)$, in agreement with
Prop. \ref{ab->conj=inv}.
\end{example}

From Prop. \ref{centraliser} and Example \ref{ex_S3_2_CDP} we know that a given set $M(E)$ may be
generated by different sets of CDPs, we have however

\begin{proposition} \label{always_gen_by_group}
Let $M(E)\equiv \{m(E_{j}):j=1,..,n\}$ be a an abelian set of matrices, so
in fact an abelian group (see Th. \ref{abelian->group}) generated by a set of CDPs $\Sigma
_{n}=\{\sigma _{i}\}_{i=1}^{n}$, which may be neither abelian nor group.
Then there exists an abelian set of CDPs $\Sigma _{n}^{\prime }=\{\sigma
_{i}^{\prime }\}_{i=1}^{n}$ such that
\begin{equation} 
M(E^{\prime })=M(E),
\end{equation}%
in particular, when $id\in \Sigma _{n}$ then $\Sigma _{n}$ is an abelian $%
CDP$ and if $id\notin \Sigma _{n}$ then $\Sigma _{n}^{\prime }=\Sigma
_{n}\sigma _{1}^{-1},$ where $\sigma _{1}\in \Sigma _{n}$ and $\sigma
_{1}(1)=1$ and such an element always exist in the set of CDPs $\Sigma _{n}.$ Thus
any abelian $M(E)$ is always generated by some abelian set of CDPs $\Sigma
_{n}^{\prime }.$
\end{proposition}

\begin{remark} \label{remark_always}
The Example \ref{ex_S3_2_CDP} illustrates this in case of the group $S(3).$
\end{remark}

The group structure of a set of CDPs allows to define a permutation, which is in natural
way connected with its group properties

\begin{definition} \label{xi}
Let $\Sigma _{n}=\{\sigma _{i}\}_{i=1}^{n}$ be a set of CDPs with $\sigma _{1}=id$%
, then we define
\begin{equation} 
\xi ^{\Sigma }\in S(n):\xi ^{\Sigma }(i)=j\Longleftrightarrow \sigma
_{i}(j)=1=\sigma _{j}(i),
\end{equation}%
which is well defined because any $\sigma _{i}\in \Sigma _{n}$ has only one
inverse. For a given permutation $\sigma _{i}\in
\Sigma _{n},$ the permutation $\xi ^{\Sigma }$ shows what is the index of
the inverse permutation $\sigma _{j}\in \Sigma _{n}$ e.i.%
\begin{equation} 
(\sigma _{i})^{-1}=\sigma _{j}=\sigma _{\xi ^{\Sigma }(i)}.
\end{equation}
\end{definition}

\begin{proposition} \label{xi_prop}
Let $\xi ^{\Sigma }\in S(n)$ be a permutation defined as above and $\Sigma
_{n}=\{\sigma _{i}\}_{i=1}^{n}$ be a set of CDPs, then

a)
\begin{equation} 
\xi ^{\Sigma }(1)=1,\qquad \xi ^{\Sigma }(i)=j\Longleftrightarrow \xi
^{\Sigma }(j)=i,\qquad \sigma _{i}\xi ^{\Sigma }(i)=1,
\end{equation}

b)%
\begin{equation} 
\xi ^{\Sigma }=\left(
\begin{array}{ccccc}
1 & . & i & . & n \\
1 & . & (\sigma _{i})^{-1}(1) & . & (\sigma _{n})^{-1}(1)%
\end{array}%
\right) =(1)(i_{1})..(i_{k})(i_{k+1}j_{k+1})..(i_{l}j_{l}),
\end{equation}%
where $i_{1},..,i_{k}$ are such that $(\sigma _{i_{p}})^{-1}=\sigma _{i_{p}}$
and $(\sigma _{i_{l}})^{-1}=\sigma _{j_{p}}.$
\end{proposition}

From these properties of the permutation $\xi ^{\Sigma }\in S(n)$ and from
Cor.\ref{abelian_to_group} one deduce easily

\begin{corollary} \label{inv->xi}
Let $\Sigma _{n}=\{\sigma _{i}\}_{i=1}^{n}$ be a set of CDPs and $%
E=\{E_{i}:i=1,..,n\}$ the conjugated set of permutations, then
\begin{equation} 
E_{i}^{-1}=E_{\xi ^{\Sigma }(i)}.
\end{equation}%
so $\ \ $the permutation $\xi ^{\Sigma }$ shows also what is the index of
the inverse permutation $E_{i}\in E$.
\end{corollary}

Let look on some examples.

\begin{example} \label{xi_for_cycl}
Let $\Sigma _{n}=C(n)=\{c^{i}=(c)^{i}:i=0,1,...,n-1\}\subset $ $S(n)$ where $%
c=(012...n-1)\Rightarrow c^{k}(t)=t+_{n}k.$ Then $(c^{k})^{-1}=c^{n-k}$ and
\begin{equation} 
\xi ^{\Sigma }(k)=n-k.
\end{equation}
\end{example}

\begin{example} \label{xi_reg_S3}
It is clear that in case of the group $S(3),$ realised as a set of CDPs in regular
representation

$\Sigma _{n}=R(S(3))=\{\sigma _{1}=id,\;\sigma _{2}=(123)(456),\;\;\sigma
_{3}=(132)(465),$

$\;\sigma _{4}=(14)(26)(53),\;\sigma _{5}=(15)(24)(36),\;\sigma
_{6}=(16)(25)(34)\},$

the permutation $\xi ^{\Sigma }$ has very simple form $\xi ^{\Sigma }=(23).$
\end{example}

\section{CDP Matrices. Generalisation of the
Circulant Matrices.}

\subsection{CDP Matrices.}

Observe, that any set of CDPs induces a decomposition of $\mathbb{C}^n \otimes \mathbb{C}^n = \mathrm{span} \{ e_i \otimes e_j \}$ into a direct sum on $n$-dimensional subspaces $\mathcal{H}_0 \oplus \mathcal{H}_1 \oplus \dots \oplus \mathcal{H}_{n-1}$, where $\mathcal{H}_k = \mathrm{span} \{ e_l \otimes e_{\sigma_k(l)}\}$. The facts, that permutations are pairwise completely different guarantees that $\mathcal{H}_k \perp \mathcal{H}_l$ for $k \ne l$.

\vskip 0.5cm

We define the following class of operators over tensor product which we will call \textit{CDP matrices}.

\begin{definition} \label{def_circ_gen}
Let $\Sigma _{n}=\{\sigma _{i}\}_{i=1}^{n}$ be a set of CDPs and $\mathcal{A}%
=\{A^{k}=(a_{ij}^{k})\in M(n,%
\mathbb{C}
):k=1,..,n\}$ a set of matrices. Then we define
\begin{equation} 
\rho \lbrack \mathcal{A},\Sigma _{n}\mathcal{]\equiv }\sum_{i,j=1}^{n}%
\sum_{k=1}^{n}a_{ij}^{k}e_{ij}\otimes e_{\sigma _{k}(i)\sigma _{k}(j)}\equiv
\sum_{k=1}^{n}\rho \lbrack A^{k},\sigma _{k}\mathcal{]\equiv }%
\sum_{i,j=1}^{n}e_{ij}\otimes B_{ij}(\mathcal{A},\Sigma _{n})\in M(n^{2},%
\mathbb{C}
),
\end{equation}%
where
\begin{equation} 
\rho \lbrack A^{k},\sigma _{k}\mathcal{]=}\sum_{i,j=1}^{n}a_{ij}^{k}e_{ij}%
\otimes e_{\sigma _{k}(i)\sigma _{k}(j)}\in M(n^{2},%
\mathbb{C}
),\qquad B_{ij}(\mathcal{A},\Sigma _{n})=\sum_{k=1}^{n}a_{ij}^{k}e_{\sigma
_{k}(i)\sigma _{k}(j)}\in M(n,%
\mathbb{C}
).
\end{equation}%
\end{definition}
Operators $\rho[A^{k},\sigma _{k}]$ are supported on $\mathcal{H}_k$. When a set of CDPs $\Sigma _{n}$ is a cyclic group $C(n)=%
\{c^{i}=(c)^{i}:i=0,1,...,n-1\}\subset $ $S(n)$ where $c=(012...n-1)$, we
recognize the definition of the circulant matrices from the paper \cite{circ}.

\begin{example}
  Take $n=4$ and consider two sets of CDPs: circulant one  $C(4)= \{ {\rm id}, c,c^2,c^3\}$, with $c=(0123)$, and $V(4) = \{ {\rm id},(01)(23),(02)(13),(03)(12)\}$. For $C(4)$ one finds the following decomposition of the total Hilbert space
\begin{equation}
  \mathbb{C}^4 \ot \mathbb{C}^4 = \mathcal{H}_0 \oplus \mathcal{H}_1 \oplus \mathcal{H}_2 \oplus \mathcal{H}_3 ,
\end{equation}
with
\begin{eqnarray*}
\mathcal{H}_0 &=& \mbox{span}\left\{ e_0 \ot e_0\, , e_1 \ot e_1\, ,e_2 \ot e_2\, , e_3 \ot e_3 \right\}
\ ,    \nonumber \\
\mathcal{H}_1 &=& \mbox{span}\left\{ e_0 \ot e_1\, , e_1\ot
e_2\, , e_2\ot e_3\, , e_3 \ot e_0 \right\}
 \    ,\\
\mathcal{H}_2 &=& \mbox{span}\left\{ e_0 \ot e_2\, , e_1 \ot
e_3\, , e_2\ot e_1\, , e_3 \ot e_2 \right\}
 \    ,\\
\mathcal{H}_3 &=& \mbox{span}\left\{ e_0 \ot e_3\, , e_1 \ot
e_0\, , e_2\ot e_1\, , e_3 \ot e_2 \right\} \    .
\end{eqnarray*}
One finds the corresponding bipartite operator

\begin{widetext}
\begin{equation}\label{A1}
  A_1 = \left( \begin{array}{cccc|cccc|cccc|cccc}
    a_{00} & \cdot  & \cdot  & \cdot  &
    \cdot  & a_{01} & \cdot  & \cdot  &
    \cdot  & \cdot  & a_{02} & \cdot  &
    \cdot  & \cdot  & \cdot  & a_{03}   \\
    \cdot  & b_{00} & \cdot  & \cdot  &
    \cdot  & \cdot  & b_{01} & \cdot  &
    \cdot  & \cdot  & \cdot  & b_{02} &
    b_{03} & \cdot  & \cdot  & \cdot    \\
    \cdot  & \cdot  & c_{00} & \cdot  &
    \cdot  & \cdot  & \cdot  & c_{01} &
    c_{02} & \cdot  & \cdot  & \cdot  &
    \cdot  & c_{03} & \cdot  & \cdot     \\
    \cdot  & \cdot  & \cdot  & d_{00} &
    d_{01} & \cdot  & \cdot  & \cdot  &
    \cdot  & d_{02} & \cdot  & \cdot  &
    \cdot  & \cdot  & d_{03} & \cdot     \\   \hline
    \cdot  & \cdot  & \cdot  & d_{10} &
    d_{11} & \cdot  & \cdot  & \cdot  &
    \cdot  & d_{12} & \cdot  & \cdot  &
    \cdot  & \cdot  & d_{13} & \cdot     \\
    a_{10} & \cdot  & \cdot  & \cdot  &
    \cdot  & a_{11} & \cdot  & \cdot  &
    \cdot  & \cdot  & a_{12} & \cdot  &
    \cdot  & \cdot  & \cdot  & a_{13}   \\
    \cdot  & b_{10} & \cdot  & \cdot  &
    \cdot  & \cdot  & b_{11} & \cdot  &
    \cdot  & \cdot  & \cdot  & b_{12} &
    b_{13} & \cdot  & \cdot  & \cdot    \\
    \cdot  & \cdot  & c_{10} & \cdot  &
    \cdot  & \cdot  & \cdot  & c_{11} &
    c_{12} & \cdot  & \cdot  & \cdot  &
    \cdot  & c_{13} & \cdot  & \cdot    \\   \hline
    \cdot  & \cdot  & c_{20} & \cdot  &
    \cdot  & \cdot  & \cdot  & c_{21} &
    c_{22} & \cdot  & \cdot  & \cdot  &
    \cdot  & c_{23} & \cdot  & \cdot    \\
    \cdot  & \cdot  & \cdot  & d_{20} &
    d_{21} & \cdot  & \cdot  & \cdot  &
    \cdot  & d_{22} & \cdot  & \cdot  &
    \cdot  & \cdot  & d_{23} & \cdot     \\
    a_{20} & \cdot  & \cdot  & \cdot  &
    \cdot  & a_{21} & \cdot  & \cdot  &
    \cdot  & \cdot  & a_{22} & \cdot  &
    \cdot  & \cdot  & \cdot  & a_{23}   \\
    \cdot  & b_{20} & \cdot  & \cdot  &
    \cdot  & \cdot  & b_{21} & \cdot  &
    \cdot  & \cdot  & \cdot  & b_{22} &
    b_{23} & \cdot  & \cdot  & \cdot    \\   \hline
    \cdot  & b_{30} & \cdot  & \cdot  &
    \cdot  & \cdot  & b_{31} & \cdot  &
    \cdot  & \cdot  & \cdot  & b_{32} &
    b_{33} & \cdot  & \cdot  & \cdot    \\
    \cdot  & \cdot  & c_{30} & \cdot  &
    \cdot  & \cdot  & \cdot  & c_{31} &
    c_{32} & \cdot  & \cdot  & \cdot  &
    \cdot  & c_{33} & \cdot  & \cdot    \\
    \cdot  & \cdot  & \cdot  & d_{30} &
    d_{31} & \cdot  & \cdot  & \cdot  &
    \cdot  & d_{32} & \cdot  & \cdot  &
    \cdot  & \cdot  & d_{33} & \cdot     \\
    a_{30} & \cdot  & \cdot  & \cdot  &
    \cdot  & a_{31} & \cdot  & \cdot  &
    \cdot  & \cdot  & a_{32} & \cdot  &
    \cdot  & \cdot  & \cdot  & a_{33}
     \end{array} \right)\ .
\end{equation}
\end{widetext}
For $V(4)$ one finds the following decomposition of the total Hilbert space
\begin{equation}
  \mathbb{C}^4 \ot \mathbb{C}^4 = \mathcal{H}'_0 \oplus \mathcal{H}'_1 \oplus \mathcal{H}'_2 \oplus \mathcal{H}'_3 ,
\end{equation}
with
\begin{eqnarray*}
\mathcal{H}'_0 &=& \mbox{span}\left\{ e_0 \ot e_0\, , e_1 \ot e_1\, ,e_2 \ot e_2\, , e_3 \ot e_3 \right\}
\ ,    \nonumber \\
\mathcal{H}'_1 &=& \mbox{span}\left\{ e_0 \ot e_1\, , e_1\ot
e_0\, , e_2\ot e_3\, , e_3 \ot e_2 \right\}
 \    ,\\
\mathcal{H}'_2 &=& \mbox{span}\left\{ e_0 \ot e_2\, , e_1 \ot
e_3\, , e_2\ot e_0\, , e_3 \ot e_1 \right\}
 \    ,\\
\mathcal{H}'_3 &=& \mbox{span}\left\{ e_0 \ot e_3\, , e_1 \ot
e_2\, , e_2\ot e_1\, , e_3 \ot e_0 \right\} \    .
\end{eqnarray*}
One finds the corresponding bipartite operator
\begin{widetext}
\begin{equation}\label{A2}
A_2\ =\ \left( \begin{array}{cccc|cccc|cccc|cccc}
    a_{00} & \cdot  & \cdot  & \cdot  &
    \cdot  & a_{01} & \cdot  & \cdot  &
    \cdot  & \cdot  & a_{02} & \cdot  &
    \cdot  & \cdot  & \cdot  & a_{03}   \\
    \cdot  & b_{00} & \cdot  & \cdot  &
    b_{01}  & \cdot  & \cdot & \cdot  &
    \cdot  & \cdot  & \cdot  & b_{02} &
    \cdot & \cdot  & b_{03}  & \cdot    \\
    \cdot  & \cdot  & c_{00} & \cdot  &
    \cdot  & \cdot  & \cdot  & c_{01} &
    c_{02} & \cdot  & \cdot  & \cdot  &
    \cdot  & c_{03} & \cdot  & \cdot     \\
    \cdot  & \cdot  & \cdot  & d_{00} &
    \cdot & \cdot  & d_{01}  & \cdot  &
    \cdot  & d_{02} & \cdot  & \cdot  &
    d_{03}  & \cdot  & \cdot & \cdot     \\   \hline
    \cdot  & b_{10} & \cdot  & \cdot  & 
    b_{11} & \cdot  & \cdot  & \cdot  &
    \cdot  &  \cdot  & \cdot  & b_{12} &
    \cdot  & \cdot  & b_{13} & \cdot     \\
    a_{10} & \cdot  & \cdot  & \cdot  &
    \cdot  & a_{11} & \cdot  & \cdot  &
    \cdot  & \cdot  & a_{12} & \cdot  &
    \cdot  & \cdot  & \cdot  & a_{13}   \\
    \cdot  &  \cdot  & \cdot  & d_{10} &
    \cdot  & \cdot  & d_{11} & \cdot  &
    \cdot  & d_{12} & \cdot  & \cdot  & 
    d_{13} & \cdot  & \cdot  & \cdot    \\
    \cdot  & \cdot  & c_{10} & \cdot  &
    \cdot  & \cdot  & \cdot  & c_{11} &
    c_{12} & \cdot  & \cdot  & \cdot  &
    \cdot  & c_{13} & \cdot  & \cdot    \\   \hline
    \cdot  & \cdot  & c_{20} & \cdot  &
    \cdot  & \cdot  & \cdot  & c_{21} &
    c_{22} & \cdot  & \cdot  & \cdot  &
    \cdot  & c_{23} & \cdot  & \cdot    \\
    \cdot  & \cdot  & \cdot  & d_{20} &
     \cdot  & \cdot  & d_{21} & \cdot  &
    \cdot  & d_{22} & \cdot  & \cdot  &
     d_{23} & \cdot  & \cdot  & \cdot     \\
    a_{20} & \cdot  & \cdot  & \cdot  &
    \cdot  & a_{21} & \cdot  & \cdot  &
    \cdot  & \cdot  & a_{22} & \cdot  &
    \cdot  & \cdot  & \cdot  & a_{23}   \\
    \cdot  & b_{20} & \cdot  & \cdot  &
    b_{21} & \cdot  & \cdot  &  \cdot  &
    \cdot  & \cdot  & \cdot  & b_{22} &
     \cdot  & \cdot  & b_{23} & \cdot    \\   \hline
    \cdot  & \cdot  & \cdot  & d_{30} &
     \cdot  & \cdot  & d_{31} & \cdot  &
    \cdot  & d_{32} & \cdot  & \cdot  &
    d_{33} & \cdot  & \cdot  &  \cdot     \\
    \cdot  & \cdot  & c_{30} & \cdot  &
    \cdot  & \cdot  & \cdot  & c_{31} &
    c_{32} & \cdot  & \cdot  & \cdot  &
    \cdot  & c_{33} & \cdot  & \cdot    \\
    \cdot  & b_{30} & \cdot  & \cdot  &
    b_{31} & \cdot  & \cdot  &  \cdot  &
    \cdot  & \cdot  & \cdot  & b_{32} &
    \cdot  & \cdot  & b_{33} & \cdot    \\
    a_{30} & \cdot  & \cdot  & \cdot  &
    \cdot  & a_{31} & \cdot  & \cdot  &
    \cdot  & \cdot  & a_{32} & \cdot  &
    \cdot  & \cdot  & \cdot  & a_{33}
     \end{array} \right)\ .
\end{equation}
\end{widetext}

\end{example}
Due to the CDP structure of the set $\Sigma _{n}=\{\sigma
_{i}\}_{i=1}^{n} $ the matrices $\rho \lbrack \mathcal{A},\Sigma _{n}%
\mathcal{]}$ have very the following spectral properties.

\begin{proposition} \label{eigen}
Let $\Sigma _{n}=\{\sigma _{i}\}_{i=1}^{n}$ be a set of CDPs and $\mathcal{A}%
=\{A^{k}=(a_{ij}^{k})\in M(n,%
\mathbb{C}
):k=1,..,n\}$ a set of arbitrary matrices, then

a) The matrix $\rho \lbrack \mathcal{A},\Sigma _{n}\mathcal{]}\equiv
\sum_{k=1}^{n}\rho \lbrack A^{k},\sigma _{k}\mathcal{]}$ is in fact a direct
sum of operators, because we have
\begin{equation} 
k\neq l\Rightarrow \rho \lbrack A^{k},\sigma _{k}\mathcal{]}\rho \lbrack
A^{l},\sigma _{l}]\mathcal{=}0,
\end{equation}%
so in particular they commute,

b) If the matrices $\mathcal{A}=\{A^{k}\in M(n,%
\mathbb{C}
):k=1,..,n\}$ are diagonalizable e.i.
\begin{equation} 
A^{k}x^{k}(q)=\lambda _{q}^{k}x^{k}(q),\qquad x^{k}(q)=(x_{j}^{k}(q))\in
\mathbb{C}
^{n},\quad q,j=1,...,n
\end{equation}%
then the matrices $\rho \lbrack A^{k},\sigma _{k}\mathcal{]}$ are also
diagonalizable and we have
\begin{equation} 
\rho \lbrack A^{k},\sigma _{k}\mathcal{]}w^{l}(q)=\delta _{kl}\lambda
_{q}^{k}w^{k}(q),
\end{equation}%
where%
\begin{equation} 
w^{k}(q)=\sum_{j=1}^{n}x_{j}^{k}(q)e_{j}\otimes e_{\sigma _{k}(j)}
\end{equation}%
is an eigenvector of the matrix $\rho \lbrack A^{k},\sigma _{k}\mathcal{]}$
corresponding to the eigenvalues $\lambda _{q}^{k}$ and all remaining
eigenvectors $w^{l}(q):l\neq k$ have eigenvalues $0.$

c) For the matrix $\rho \lbrack \mathcal{A},\Sigma _{n}\mathcal{]}$ we have
the following eigen-equation
\begin{equation} 
\rho \lbrack \mathcal{A},\Sigma _{n}\mathcal{]}w^{k}(q)=\lambda
_{q}^{k}w^{k}(q),
\end{equation}%
so eigenvalues of the matrices $\mathcal{A}=\{A^{k}\in M(n,%
\mathbb{C}
):k=1,..,n\}$ are eigenvalues of the matrix $\rho \lbrack \mathcal{A}%
,\Sigma _{n}\mathcal{]}$.
\end{proposition}

From this structure of CDP matrices $\rho \lbrack \mathcal{A}%
,\Sigma _{n}\mathcal{]}$ one easily deduce \ the following norm properties

\begin{proposition} \label{norms}
Let $\Sigma _{n}=\{\sigma _{i}\}_{i=1}^{n}$ be a set of CDPs and $%
\mathcal{A}=\{A^{k}=(a_{ij}^{k})\in M(n,%
\mathbb{C}
):k=1,..,n\}$ a set of arbitrary matrices, then
\begin{equation} 
||\rho \lbrack \mathcal{A},\Sigma _{n}\mathcal{]}||_{HS}=%
\sum_{k=1}^{n}||A^{k}||_{HS},\qquad ||\rho \lbrack \mathcal{A},\Sigma _{n}%
\mathcal{]}||_{tr}=\sum_{k=1}^{n}||A^{k}||_{tr},
\end{equation}%
where%
\begin{equation} 
||X||_{HS}=[trX^{+}X]^{\frac{1}{2}},\qquad ||X||_{tr}=tr[X^{+}X]^{\frac{1}{2}%
},
\end{equation}%
so these norms of $\rho \lbrack \mathcal{A},\Sigma _{n}\mathcal{]}$ depends
directly on the norms of matrices only and does not depends on a set of CDPs $%
\Sigma _{n}=\{\sigma _{i}\}_{i=1}^{n}.$
\end{proposition}

The matrix $\rho \lbrack \mathcal{A},\Sigma _{n}\mathcal{]}$ may be written
also using the set of matrices $M(E)\equiv \{m(E_{j}):j=1,..,n\}$, defined
in Def. \ref{mE}, namely we have

\begin{proposition} \label{similarCHK}
Under assumptions of the Prop. \ref{eigen} the block structure of the matrix $\rho
\lbrack \mathcal{A},\Sigma _{n}\mathcal{]}$ is the following
\begin{equation} 
\rho \lbrack \mathcal{A},\Sigma _{n}\mathcal{]=}\sum_{i,j=1}^{n}e_{ij}%
\otimes B_{ij}(\mathcal{A},\Sigma _{n})\in M(n^{2},%
\mathbb{C}
),
\end{equation}%
where
\begin{equation} 
B_{ij}(\mathcal{A},\Sigma _{n})=\sum_{k=1}^{n}a_{ij}^{k}e_{\sigma
_{k}(i)\sigma _{k}(j)}=m(E_{i})^{+}A_{ij}m(E_{j}),
\end{equation}%
so, we have
\begin{equation} 
\rho \lbrack \mathcal{A},\Sigma _{n}\mathcal{]=}\sum_{i,j=1}^{n}e_{ij}%
\otimes m(E_{i})^{+}A_{ij}m(E_{j}):A_{ij}=(a_{ij}^{k}\delta _{kl}),
\end{equation}%
which is similar to formula in \cite{circ}.
\end{proposition}

\begin{definition} \label{ab_tp_matr}
The CDP matrices
\begin{equation} 
\rho \lbrack \mathcal{A},\Sigma _{n}\mathcal{]=}\sum_{i,j=1}^{n}e_{ij}%
\otimes m(E_{i})^{+}A_{ij}m(E_{j}):A_{ij}=(a_{ij}^{k}\delta _{kl}),
\end{equation}%
where the set $M(E)\equiv \{m(E_{j}):j=1,..,n\}$ is abelian (equivalently the $\Sigma_n=\{\sigma_k\}$ is abelian) we will call
abelian or commutative CDP matrices.
\end{definition}

From this structure of the matrix $\rho \lbrack \mathcal{A},\Sigma _{n}%
\mathcal{]}$ we deduce easily that

\begin{proposition} \label{herm_semipos}
Let $\Sigma _{n}=\{\sigma _{i}\}_{i=1}^{n}$ be a set of CDPs and $\mathcal{A}%
=\{A^{k}\in M(n,%
\mathbb{C}
):k=1,..,n\}$ a set of matrices. Then the matrix $\rho \lbrack \mathcal{A}%
,\Sigma _{n}\mathcal{]}$ is hermitian iff the matrices $\mathcal{A}$ are
hermitian and similarly the matrix $\rho \lbrack \mathcal{A},\Sigma _{n}%
\mathcal{]}$ is semipositive definite iff all matrices $\mathcal{A}$ are
semi-positive definite.
\end{proposition}

Unfortunately the good properties of the matrices $\rho \lbrack \mathcal{A}%
,\Sigma _{n}\mathcal{]}$ disappear after partial transpose of second part of
the tensor product, in fact we have

\begin{proposition} \label{after_pt}
Let $\Sigma _{n}=\{\sigma _{i}\}_{i=1}^{n}$ be a set of CDPs and $\mathcal{A}%
=\{A^{k}=(a_{ij}^{k})\in M(n,%
\mathbb{C}
):k=1,..,n\}$ be a set of hermitian matrices, then a partial transposed
matrices $\rho \lbrack A^{k},\sigma _{k}\mathcal{]}^{T_{2}}=id\otimes T \rho
\lbrack A^{k},\sigma _{k}\mathcal{]}$ where $T$ is the transpose in $M(n,%
\mathbb{C}
)$ have the following properties

a) the matrices $\rho \lbrack A^{k},\sigma _{k}\mathcal{]}$ are hermitian
but in general they do not commute,

b) the matrix $\rho \lbrack A^{k},\sigma _{k}\mathcal{]}^{T_{2}}$ has $%
n^{2}-n$ eigenvectors
\begin{equation} 
w_{pq}^{k}=a_{p\sigma _{k}^{-1}(q)}^{k}e_{p}\otimes e_{q}\pm |a_{p\sigma
_{k}^{-1}(q)}^{k}|e_{\sigma _{k}^{-1}(q)}\otimes e_{\sigma _{k}(p)}:p\neq
\sigma _{k}^{-1}(q),
\end{equation}%
corresponding to the eigenvalues
\begin{equation} \label{75}
\gamma _{pq}^{k}=\pm |a_{p\sigma _{k}^{-1}(q)}^{k}|,
\end{equation}%
and $n$ eigenvectors and eigenvalues of the form
\begin{equation} 
w_{pp}^{k}=e_{p}\otimes e_{\sigma _{k}(p)},\qquad \gamma _{pp}^{k}=a_{pp}.
\end{equation}
\end{proposition}

What can be checked by a direct calculation. From this it follows that

\begin{corollary} \label{sep_prop}
If $\Sigma _{n}=\{\sigma _{i}\}_{i=1}^{n}$ be a set of CDPs and $\mathcal{A}%
=\{A^{k}\in M(n,%
\mathbb{C}
):k=1,..,n\}$ is a set of hermitian matrices, then

a) if the matrix $A^{k}$ is not diagonal, then a partially transposed
matrices $\rho \lbrack A^{k},\sigma _{k}\mathcal{]}^{T_{2}}$ is not
semi-positive definite and therefore the matrix $\rho \lbrack A^{k},\sigma
_{k}\mathcal{]}$ is $NPT$ state,

b) if the matrices $\mathcal{A}=\{A^{k}\in M(n,%
\mathbb{C}
):k=1,..,n\}$ are not diagonal, then the matrix $\rho \lbrack \mathcal{A}%
,\Sigma _{n}\mathcal{]}=\sum_{k=1}^{n}\rho \lbrack A^{k},\sigma _{k}\mathcal{%
]}$ is a direct sum of $NPT$ matrices.

c) Two unitarily equivalent matrices $A^{k},\quad A^{\prime
k}=UA^{k}U^{+}:U\in U(n)$ may define CDP matrices $\rho \lbrack
A^{k},\sigma _{k}\mathcal{]}$ and $\rho \lbrack A^{\prime k},\sigma _{k}%
\mathcal{]}$ with non-equivalent separability properties. For example if $%
A^{k} $ is a non-diagonal matrix and $A^{\prime k}=UA^{k}U^{+}$ is a diagonal
form of the matrix $A^{k}$, then $\rho \lbrack A^{\prime k},\sigma _{k}%
\mathcal{]} $ is a $PPT$ matrix, whereas $\rho \lbrack A^{k},\sigma _{k}%
\mathcal{]}$ is a $NPT$ matrix.
\end{corollary}

\begin{remark} \label{eigenv_after_pt}
In general the matrix eigenvalues depends in a complicated way on matrix
elements and generally there is no formulae that describe this dependence.
In particular the eigen values $\lambda _{q}^{k}$ of the matrices $\rho
\lbrack A,\Sigma _{n}]$ depends on the matrix elements $(a_{ij}^{k})$ in
such a complicated, in general way. It remarcable that after partial
transpose of the matrix  the eigenvalues of matrix $\rho \lbrack A,\Sigma
_{n}]^{T_{2}}$ depends in a very simple way given by Eq. (\ref{75}).
\end{remark}

We have one more useful property of the matrix $\rho \lbrack \mathcal{A}%
,\Sigma _{n}\mathcal{]}$

\begin{proposition} \label{induced_by_nat_trans}
Let $\Sigma _{n}=\{\sigma _{i}\}_{i=1}^{n}$ be a set of CDPs and $\Sigma
_{n}^{\prime }=\{\sigma _{i}^{\prime }=\delta \sigma _{i}\eta \}_{i=1}^{n}$%
,\ where $\delta ,\eta \in S(n)$, so $\Sigma _{n}^{\prime }$ is also a set of CDPs
and $\mathcal{A}=\{A^{k}=(a_{ij}^{k})\in M(n,%
\mathbb{C}
):k=1,..,n\}$ a set of matrices. Then
\begin{equation} 
\rho \lbrack \mathcal{A},\Sigma _{n}^{\prime }\mathcal{]\equiv }\rho \lbrack
A^{k},\Sigma _{n}^{\prime }\mathcal{]=}m(\eta ^{-1})\otimes m(\delta )\rho
\lbrack m(\eta ^{-1})A^{k}m(\eta ),\Sigma _{n}\mathcal{]}m(\eta )\otimes
m(\delta ^{-1}),
\end{equation}%
so, the elementary transformation of a set of CDPs $\Sigma _{n}\rightarrow $ $%
\Sigma _{n}^{\prime }$ induces a local unitary transformations of the matrix
$\rho \lbrack \mathcal{A},\Sigma _{n}^{\prime }\mathcal{]\equiv }\rho
\lbrack A^{k},\Sigma _{n}^{\prime }\mathcal{]}$ together with a similarity
transformation of the matrices $A^{k}\rightarrow m(\eta ^{-1})A^{k}m(\eta )$.
\end{proposition}

This Proposition and the Proposition \ref{always_gen_by_group} give us the following statement
concerning CDP matrices $\rho \lbrack A^{k},\Sigma
_{n}\mathcal{]}.$

\begin{corollary} \label{ab_tp_matr<-gcdp}
Suppose that a commutative CDP matrix%
\begin{equation} 
\rho \lbrack A^{k},\Sigma _{n}\mathcal{]=}\sum_{i,j=1}^{n}e_{ij}\otimes
m(E_{i})^{+}A_{ij}m(E_{j}):A_{ij}=(a_{ij}^{k}\delta _{kl})
\end{equation}%
is generated by a set of CDPs $\Sigma _{n}=\{\sigma _{i}\}_{i=1}^{n}$ (so the set
of matrices $M(E)\equiv \{m(E_{j}):j=1,..,n\}$ is abelian), than there exist
an abelian set of CDPs $\Sigma _{n}^{\prime }=\{\sigma _{i}^{\prime
}\}_{i=1}^{n}=\Sigma _{n}\sigma _{1}^{-1}:\sigma _{1}\in \Sigma _{n},\quad
\sigma _{1}(1)=1$ such that
\begin{equation} 
\rho \lbrack A^{k},\Sigma _{n}\mathcal{]=}m(\sigma _{1}^{-1})\otimes id\rho
\lbrack m(\sigma _{1}^{-1})A^{k}m(\sigma _{1}),\Sigma _{n}^{\prime }\mathcal{%
]}m(\sigma _{1})\otimes id,
\end{equation}%
so the commutative CDP matrices $\rho \lbrack A^{k},\Sigma _{n}%
\mathcal{]}$ are in fact generated by abelian set of CDPs.
\end{corollary}

\subsection{The Partial Transpose of CDP
Matrices.}

Let us consider the partial transpose of the CDP matrix $\rho \lbrack
A,\Sigma _{n}]^{T_{2}}$. It is clear that, in general the matrix $\rho
\lbrack A,\Sigma _{n}]^{T_{2}}$ has different structure in comparison with the
matrix  $\rho \lbrack A,\Sigma _{n}].$ It appears however that if the matrix
$\rho \lbrack A,\Sigma _{n}]$ is commutative e.i. if a set of CDP $\Sigma
_{n}=\{\sigma _{i}\}$ is abelian, then the matrix $\rho \lbrack A,\Sigma
_{n}]^{T_{2}}$ is also a CDP matrix from the set of CDPs e.i. $\rho \lbrack
A,\Sigma _{n}]^{T_{2}}=\rho \lbrack \widetilde{A},\Sigma _{n}^{\prime }],$
where $\Sigma _{n}^{\prime }$ is a set of $CDP.$ In fact from Prop. \ref{similarCHK} we
have
\begin{equation} 
\rho \lbrack A,\Sigma _{n}]=\sum_{i,j=1}^{n}e_{ij}\otimes
\sum_{k=1}^{n}a_{ij}^{k}e_{\sigma _{k}(i)\sigma _{k}(j)}\Rightarrow \rho
\lbrack A,\Sigma _{n}]^{T_{2}}=\sum_{i,j=1}^{n}e_{ij}\otimes
\sum_{k=1}^{n}a_{ij}^{k}e_{\sigma _{k}(j)\sigma _{k}(i)}.
\end{equation}%
On the other hand let us consider the matrix
\begin{equation} 
\rho \lbrack \widetilde{A},\Sigma _{n}\xi ^{\Sigma }]:\widetilde{A}^{k}=(%
\widetilde{a}_{ij}^{k})\equiv (a_{ij}^{\sigma _{i}^{-1}\sigma _{j}^{-1}(k)}),
\end{equation}%
where the permutation $\xi ^{\Sigma }$ is defined in Def. \ref{xi}. From Prop. \ref{similarCHK}
we have
\begin{equation} 
\rho \lbrack \widetilde{A},\Sigma _{n}\xi ^{\Sigma
}]=\sum_{i,j=1}^{n}e_{ij}\otimes \sum_{k=1}^{n}\widetilde{a}%
_{ij}^{k}e_{\sigma _{k}\xi ^{\Sigma }(j)\sigma _{k}\xi ^{\Sigma
}(i)}=\sum_{i,j=1}^{n}e_{ij}\otimes \sum_{k=1}^{n}\widetilde{a}%
_{ij}^{k}e_{\sigma _{\xi ^{\Sigma }(i)}(k)\sigma _{\xi ^{\Sigma }(j)}(k)},
\end{equation}%
where in the last step we have used the commutativity of the set of CDPs $\Sigma
_{n}=\{\sigma _{i}\}.$ Next using the definition of $\widetilde{A}^{k}=(%
\widetilde{a}_{ij}^{k})$ and Def. \ref{xi} we get
\begin{equation} 
\rho \lbrack \widetilde{A},\Sigma _{n}\xi ^{\Sigma
}]=\sum_{i,j=1}^{n}e_{ij}\otimes \sum_{k=1}^{n}a_{ij}^{\sigma
_{i}^{-1}\sigma _{j}^{-1}(k)}e_{\sigma _{(i)}^{-1}(k)\sigma _{(j)}^{-1}(k)}.
\end{equation}%
Making substitution $l=\sigma _{i}^{-1}\sigma _{j}^{-1}(k),$ we obtain
\begin{equation} 
\rho \lbrack \widetilde{A},\Sigma _{n}\xi ^{\Sigma
}]=\sum_{i,j=1}^{n}e_{ij}\otimes \sum_{l=1}^{n}a_{ij}^{l}e_{\sigma
_{j}(l)\sigma _{i}(l)}=\sum_{i,j=1}^{n}e_{ij}\otimes
\sum_{l=1}^{n}a_{ij}^{l}e_{\sigma _{l}(j)\sigma _{l}(i)}=\rho \lbrack
A,\Sigma _{n}]^{T_{2}}.
\end{equation}%
So we may formulate the main result of this section

\begin{theorem} \label{iff_ppt}
Suppose that a CDP matrix $\rho \lbrack \mathcal{A},\Sigma _{n}\mathcal{%
]}$ is a commutative and is generated by an abelian set of CDPs $\Sigma
_{n}=\{\sigma _{i}\}_{i=1}^{n}$ and the matrices $\mathcal{A}=\{A^{k}\in M(n,%
\mathbb{C}
):k=1,..,n\}$ are semi-positive definite, then%
\begin{equation} 
\rho \lbrack \mathcal{A},\Sigma _{n}\mathcal{]}^{T_{2}}=\rho \lbrack
\widetilde{\mathcal{A}},\Sigma _{n}\xi ^{\Sigma }\mathcal{]}\quad :%
\widetilde{A}^{k}=(\widetilde{a}%
_{ij}^{k})=(a_{ij}^{E_{i}E_{j}(k)})=(a_{ij}^{\sigma _{i}^{-1}\sigma
_{j}^{-1}(k)}),
\end{equation}%
where $\xi ^{\Sigma }$ is defined in Def. \ref{xi}. The matrix $\rho \lbrack
\mathcal{A},\Sigma _{n}\mathcal{]}$ is then a $PPT$ state iff all matrices $%
\widetilde{A}^{k}$ are semi-definite.
\end{theorem}

This theorem is a generalisation of corresponding result concerning cyclic
groups \cite{circ}, to arbitrary abelian groups.

\begin{remark}
Note that for the abelian group $\Sigma _{4}=V(4)$ Example 3, for which $%
\xi ^{\Sigma }=id$ we have
\begin{equation} 
\rho \lbrack \mathcal{A},V(4)\mathcal{]}^{T_{2}}=\rho \lbrack \widetilde{%
\mathcal{A}},V(4)\mathcal{]}\quad :\widetilde{A}^{k}=(\widetilde{a}%
_{ij}^{k})=(a_{ij}^{\sigma _{i}\sigma _{j}(k)}),
\end{equation}%
so in this case, after partial transpose, the group remains the same.
\end{remark}

\subsection{Realignment Criterion for CDP
Matrices.}

It appears that matrices of the form $\rho \lbrack \mathcal{A},\Sigma _{n}%
\mathcal{]}$ are friendly for Realignment Criterion derived in \cite{RR}.
Namely we have

\begin{theorem}
Suppose that $\Sigma _{n}=\{\sigma _{i}\}_{i=1}^{n}$ is an abelian set of CDPs ,
then%
\begin{equation} 
\rho \lbrack \mathcal{A},\Sigma _{n}\mathcal{]}^{RL}=\rho \lbrack \widetilde{%
\mathcal{A}},\Sigma _{n}\mathcal{]}\quad :\widetilde{A}^{k}=(\widetilde{a}%
_{ij}^{k})=(a_{\sigma _{k}(i)i}^{\sigma _{i}^{-1}(j)}),
\end{equation}%
where $RL$ means Realignement. So after this transformation the set of CDPs $%
\Sigma _{n}$ in matrix $\rho \lbrack \mathcal{A},\Sigma _{n}\mathcal{]}^{RL}$%
remains the same.
\end{theorem}

From this theorem, from Proposition 16 and Realignment Criterion it follows
immediately the following necessary condition for separability of
commutative matrices $\rho \lbrack \mathcal{A},\Sigma _{n}\mathcal{]}$.

\begin{proposition}
Let $\rho \lbrack \mathcal{A},\Sigma _{n}\mathcal{]}$ be a commutative
CDP matrix (e.i. $\Sigma _{n}=\{\sigma _{i}\}_{i=1}^{n}$ is an
abelian) then $\rho \lbrack \mathcal{A},\Sigma _{n}\mathcal{]}$ may be
separable only if%
\begin{equation} 
\sum_{k=1}^{n}||\widetilde{A}^{k}||_{tr}\leq 1.
\end{equation}
\end{proposition}

\subsection{Majorisation Criterion for Tensor matrices from sets of
CDPs.}

The matrices $\rho \lbrack \mathcal{A},\Sigma _{n}\mathcal{]}$ have, by
construction, traceless off diagonal blocks and the diagonal blocks are
diagonal, which may imply that a majorisation criteria of entanglement
could be easier to application for such a matrices. We will need

\begin{definition}
Let $A\in M(n,%
\mathbb{C}
)$ and $A^{+}=A.$ Then $\lambda (A)\in
\mathbb{C}
^{n}$ is the vector whose components are the eigenvalues of $A$ arranged in
decreasing order e.i. we have%
\begin{equation} 
\lambda (A)=(\lambda _{k}(A)),\qquad \lambda _{1}(A)\geq \lambda _{2}(A)\geq
...\geq \lambda _{n}(A).
\end{equation}%
We say that a matrix $A$ is majorised by a matrix $B$, which is denoted $%
A\prec B$ if
\begin{equation} 
\lambda _{1}(A)\leq \lambda _{1}(B),
\end{equation}%
\begin{equation} 
\lambda _{1}(A)+\lambda _{2}(A)\leq \lambda _{1}(B)+\lambda _{2}(B),
\end{equation}%
\begin{equation} 
..
\end{equation}%
\begin{equation} 
\lambda _{1}(A)+\lambda _{2}(A)+..+\lambda _{n-1}(A)\leq \lambda
_{1}(B)+\lambda _{2}(B)+...+\lambda _{n-1}(B),
\end{equation}%
\begin{equation} 
\lambda _{1}(A)+\lambda _{2}(A)+..+\lambda _{n}(A)=\lambda _{1}(B)+\lambda
_{2}(B)+...+\lambda _{n}(B),
\end{equation}%
e.i. if $\lambda (A)\prec \lambda (B)$, so majorisation of the hermitian
matrices is defined as majorisation of its vectors of eigenvalues.
\end{definition}

Now let us consider an arbitrary CDP matrix $\rho \lbrack \mathcal{A}%
,\Sigma _{n}\mathcal{]}$ where $\Sigma _{n}=\{\sigma _{i}\}_{i=1}^{n}$ is an
arbitrary set of CDPs and $\mathcal{A}=\{A^{k}=(a_{ij}^{k})\in M(n,%
\mathbb{C}
):k=1,..,n\}$ a set of hermitian positive matrices, then

\begin{proposition}
\begin{equation} 
\rho _{1}[\mathcal{A},\Sigma _{n}\mathcal{]\equiv }(id\otimes tr)\rho
\lbrack \mathcal{A},\Sigma _{n}\mathcal{]=}\sum_{i=1}^{n}e_{ii}(%
\sum_{k=1}^{n}a_{ii}^{k}),
\end{equation}%
\begin{equation} 
\rho _{2}[\mathcal{A},\Sigma _{n}\mathcal{]\equiv }(tr\otimes id)\rho
\lbrack \mathcal{A},\Sigma _{n}\mathcal{]=}\sum_{i=1}^{n}(%
\sum_{k=1}^{n}a_{ii}^{k}e_{\sigma _{k}(i)\sigma _{k}(i)}),
\end{equation}%
so both these matrices are diagonal with positive entries on the diagonal.
Note that $\rho _{1}[\mathcal{A},\Sigma _{n}\mathcal{]}$ depends on the
matrices $\mathcal{A}=\{A^{k}=(a_{ij}^{k})\in M(n,%
\mathbb{C}
):k=1,..,n\}$ \ only (in fact on their diagonals) and not on the set of CDPs $\Sigma
_{n}.$
\end{proposition}

Now we have the following Majorisation Criterion $(MC)$

\begin{theorem}
If a CDP state $\rho _{12\text{ }}$ is separable then
\begin{equation} 
\rho _{12}\prec \rho _{1}\mathcal{\equiv }(id\otimes tr)\rho _{12}\quad
\wedge \quad \rho _{12}\prec \rho _{2}\mathcal{\equiv }(tr\otimes id)\rho
_{12}.
\end{equation}
\end{theorem}

From this we get

\begin{proposition}
The CDP matrix $\rho \lbrack \mathcal{A},\Sigma _{n}\mathcal{]}$ where $%
\Sigma _{n}=\{\sigma _{i}\}_{i=1}^{n}$ is an arbitrary set of CDPs and $\mathcal{A%
}=\{A^{k}=(a_{ij}^{k})\in M(n,%
\mathbb{C}
):k=1,..,n\}$ a set of hermitian positive matrices may be separable only if
\begin{equation} 
\rho \lbrack \mathcal{A},\Sigma _{n}\mathcal{]\prec }\sum_{i=1}^{n}e_{ii}(%
\sum_{k=1}^{n}a_{ii}^{k})\quad \wedge \quad \rho \lbrack \mathcal{A},\Sigma
_{n}\mathcal{]\prec }\sum_{i=1}^{n}(\sum_{k=1}^{n}a_{ii}^{k}e_{\sigma
_{k}(i)\sigma _{k}(i)}),
\end{equation}%
where \ the matrices on RHS are diagonal e.i. their eigenvalues are given
explicitely so it simplify calculation of majorisation.
\end{proposition}

From the above Definition of majorisation of hermitian matrices we know that
it is in fact majorisation of corresponding vectors of eigenvalues. From
Prop. \ref{eigen} we know also that the eigenvalues of the matrix $\rho \lbrack
\mathcal{A},\Sigma _{n}\mathcal{]}$ are exactly the eigenvalues of the
matrices $\mathcal{A}=\{A^{k}=(a_{ij}^{k})\in M(n,%
\mathbb{C}
):k=1,..,n\}.$ So the relations in the last Proposition show in what a way
the eigenvalues of the matrices $\mathcal{A}$ should be majorised by sums of
their diagonal elements. On the other hand we have famous theorem by Shurr

\begin{theorem}
Let $A\in M(n,%
\mathbb{C}
)$ be a hermitian matrix, then
\begin{equation} 
d(A)\prec \lambda (A),
\end{equation}%
where $d(A)\in
\mathbb{C}
^{n}$ is the vector whose components are diagonal elements of $A$ arranged
in decreasing order.
\end{theorem}

Thus we see that the majorisation necessary conditions for separability of
the matrix $\rho \lbrack \mathcal{A},\Sigma _{n}\mathcal{]}$ from the last
Proposition give the majorisation of the eigenvalues of the matrices $%
\mathcal{A}=\{A^{k}\in M(n,%
\mathbb{C}
):k=1,..,n\}$ by sums of its diagonal elements and on the other hand, from
Schur Theorem, the eigenvalues of the matrices $\mathcal{A}$ majorises its
diagonal elements. So we see that we have non-trivial conditions for
separability for the matrices $\rho \lbrack \mathcal{A},\Sigma _{n}\mathcal{]%
}$ and in this case for arbitrary set of CDPs $\Sigma _{n}=\{\sigma
_{i}\}_{i=1}^{n}$ , not only for groups.

\section{Examples of linear maps related to sets of CDPs}

In the paper \cite{MM} (see also \cite{Kasia})  the Irreducible Covariant Quantum Channels were
introduced, which are defined in the following way.

\begin{definition}
Let
\[
u:G\rightarrow M(n,%
\mathbb{C}
),\qquad u(g)=(u_{ij}(g))\in M(n,%
\mathbb{C}
)
\]%
be an unitary irreducible representation ($IRREP$) of a given\ finite group $%
G.$ A quantum channel $\Phi ,$which is by definition completely positive and
trace preserving map is called irreducible and invariant ($ICQC$) with
respect to $IRREP$ $U:G\rightarrow M(n,%
\mathbb{C}
)$ if
\[
\forall g\in G\quad \forall x\in M(n,%
\mathbb{C}
)\quad Ad_{U(g)}[\Phi (x)]=\Phi \lbrack Ad_{U(g)}(x)],
\]%
where
\[
Ad_{U(g)}(x)\equiv U(g)xU^{+}(g),
\]%
so $\Phi $ commute with $Ad_{U(g)}$.
\end{definition}

It has been shown that under, assumption that the tensor product is simply
reducible e.i. $U\otimes \overline{U}=$ $\sum_{\alpha \in \widehat{G}%
}m_{\alpha }\varphi ^{\alpha }:m_{\alpha }=0,1,$ $\widehat{G}$ is the set of
all $IRREP^{\prime }s$, the  ($ICQC$) have the following structure

\begin{proposition}
A quantum channel $\Phi \in End[M(n,%
\mathbb{C}
)]$, which is irreducible and invariant with respect to $IRREP$ $%
U:G\rightarrow M(n,%
\mathbb{C}
)$ is necessarily of the form
\[
\Phi =l_{id}\Pi ^{id}+\sum_{\alpha \in \Theta ,\alpha \neq id}l_{\alpha }\Pi
^{\alpha }:\quad l_{\alpha }\in
\mathbb{C}
,
\]%
where
\[
\Pi ^{\alpha }=\frac{\dim \varphi ^{\alpha }}{|G|}\sum_{g\in G}\chi ^{\alpha
}(g^{-1})Ad_{U(g)}:\chi ^{\alpha }(g^{-1})=tr\varphi ^{\alpha }(g^{-1}).
\]
\end{proposition}

It appears that the the value of Choi-Jamiolkowski isomorphism on ($ICQC$)
for  $S(3)$ and quaternion groups have the structure of a CDP matrix. In
fact we have

\begin{example}
For the group $S(3)$ we have
\[
\Phi =\Pi ^{id}+l_{sgn}\Pi ^{sgn}+l_{\lambda }\Pi ^{\lambda },
\]%
where $\lambda $ denotes the two-dimensional $IRREP$ of $S(3).$ The
corresponding Choi-Jamiolkowski matrix is of the form
\[
J(\Phi )=
\left[
\begin{array}{cccc}
  \frac{1}{2}(1+l_{\mathrm{sgn}}) & 0 & 0 & l_{\lambda } \\
  0 & \frac{1}{2}(1-l_{\mathrm{sgn}}) & 0 & 0 \\
  0 & 0 & \frac{1}{2}(1-l_{\mathrm{sgn}}) & 0 \\
  l_{\lambda } & 0 & 0 & \frac{1}{2}(1+l_{\mathrm{sgn}})
\end{array}
\right]
=\rho \lbrack \mathcal{A},\Sigma _{2}],
\]%
where $\Sigma _{2}=S(2)$ and
\[
\mathcal{A}=\{A^{1},A^{2}\}:A^{1}=\left(
\begin{array}{cc}
\frac{1}{2}(1+l_{sgn}) & l_{\lambda } \\
l_{\lambda } & \frac{1}{2}(1+l_{sgn})%
\end{array}%
\right) ,\quad A^{2}=\left(
\begin{array}{cc}
\frac{1}{2}(1-l_{sgn}) & 0 \\
0 & \frac{1}{2}(1-l_{sgn})%
\end{array}%
\right) .
\]
\end{example}

\begin{example}
The quaternion group $Q=\left\lbrace \pm Q_{\operatorname{e}},\pm Q_1, \pm Q_2, \pm Q_3 \right\rbrace $ is a non-abelian group of order eight satisfying
\begin{equation}
Q=\left\langle -Q_{\operatorname{e}},Q_1,Q_2,Q_3 \ | \ \left(-Q_{\operatorname{e}} \right)^2=Q_{\operatorname{e}}, Q_1^2=Q_2^2=Q_3^2=Q_1Q_2Q_3=-Q_{\operatorname{e}} \right\rangle.
\end{equation}
It possesses five inequivalent irreducible representations which we label by $\mathrm{id},t_1,t_2,t_3,t_4$, respectively. However, only one of them, labeled by $t_4$, has dimension greater than one and its dimension is equal to two.
It is known that the quaternion group can be represented as a subgroup of $GL(2,\mathbb{C})$. The matrix representation $R:Q\rightarrow GL(2,\mathbb{C})$ is given by
\begin{equation}
Q_{\operatorname{e}}=\begin{pmatrix}
	1 & 0 \\ 0 & 1
\end{pmatrix}, \  Q_1=\begin{pmatrix}
\operatorname{i} & 0\\ 0 & -\operatorname{i}
\end{pmatrix}, \ Q_2=\begin{pmatrix}
0 & 1\\ -1 & 0
\end{pmatrix}, \ Q_3=\begin{pmatrix}
0 & \operatorname{i} \\ \operatorname{i} & 0
\end{pmatrix},
\end{equation}
where $\operatorname{i}^2=-1$. In Table \ref{t} we present values of the characters for all irreducible representations of the group $Q$.

\begin{table}[h]
	\centering
	\begin{tabular}{|c|c|c|c|c|c|c|c|c|}
		\hline
		$Q$ & $Q_{\operatorname{e}}$ & $-Q_{\operatorname{e}}$ & $Q_1$ & $Q_2$ & $Q_3$ & $-Q_1$ & $-Q_2$ & $-Q_3$\\
		\hline
		$\chi^{\mathrm{id}}$ & 1 & 1 & 1 & 1 & 1 & 1 & 1 & 1\\
		\hline
		$\chi^{t_1}$ & 1 & 1 & -1 & 1 & -1 & -1 & 1 & -1\\
		\hline
		$\chi^{t_2}$ & 1 & 1 & 1 & -1 & -1 & 1 & -1 & -1\\
		\hline
		$\chi^{t_3}$ & 1 & 1 & -1 & -1 & 1 & -1 & -1 & 1\\
		\hline
		$\chi^{t_4}$ & 2 & -2 & 0 & 0 & 0 & 0 & 0 & 0\\
		\hline
	\end{tabular}
	\caption{Table of characters for the quaternion group $Q$.}
	\label{t}
\end{table}

\begin{equation} 
  \Phi^{t_4}=l_{t_{\mathrm{id}}}\Pi^{\mathrm{id}}+l_{t_1}\Pi^{t_1}+l_{t_2}\Pi^{t_2}+l_{t_3}\Pi^{t_3}
\end{equation}

The corresponding Choi-Jamiolkowski matrix is of the form
\begin{equation}
\label{uuu}
J\left(\Phi^{t_4} \right)=
\frac{1}{2}\begin{pmatrix}
	1+l_{t_2} & 0 & 0 & l_{t_1}+l_{t_3}\\
	0 & 1-l_{t_2} & l_{t_3}-l_{t_1} & 0\\
	0 & l_{t_3}-l_{t_1} & 1-l_{t_2} & 0\\
	l_{t_1}+l_{t_3} & 0 & 0 & 1+l_{t_2}
\end{pmatrix}
= \rho\left[\mathcal{A},\Sigma_2\right],
\end{equation}
where $\Sigma _{2}=S(2)$ and
\[
\mathcal{A}=\{A^{1},A^{2}\}:A^{1}=\left(
\begin{array}{cc}
  \frac{1}{2}(1+l_{t_2}) &
  \frac 12 (l_{t_1} + l_{t_3}) \\
  \frac 12 (l_{t_1} + l_{t_3}) &
  \frac{1}{2}(1+l_{t_2})
\end{array}%
\right) ,\quad A^{2}=\left(
\begin{array}{cc}
  \frac{1}{2}(1-l_{t_2}) &
  \frac 12 (l_{t_3} - l_{t_1}) \\
  \frac 12 (l_{t_3} - l_{t_1}) &
  \frac{1}{2}(1-l_{t_2})
\end{array}%
\right) .
\]
\end{example}

Now we consider the reduction map and its generalisation to the Breuer-Hall map. It is not difficult to check that the reduction map is related to CDP matrices in the following simple and non unique way.
\begin{proposition}
Let us consider the reduction map%
\begin{equation} 
R:M(n,%
\mathbb{C}
)\rightarrow M(n,%
\mathbb{C}
);\quad R(A)=tr(A)id_{n}-A,\quad A\in M(n,%
\mathbb{C}
).
\end{equation}%
Then we have
\begin{equation} 
R^{\otimes }\equiv \sum_{ij}R(e_{ij})\otimes e_{ij}=\rho \lbrack \mathcal{A}%
,\Sigma _{n}\mathcal{]},
\end{equation}%
where
\begin{equation} 
\mathcal{A}=\{A^{1}=id_{n}-J,\quad A^{k}=id_{n}:k=2,..,n\}
\end{equation}%
and $\Sigma _{n}=\{\sigma _{i}\}_{i=1}^{n}$ is an \textbf{arbitrary} set of CDPs
such that $\sigma _{1}=id.$ So in fact, due to the structure of the matrices
$\mathcal{A}$, the reduction map weakly depends on the set of CDPs $\Sigma_{n}=\{\sigma _{i}\}_{i=1}^{n}$.
\end{proposition}

Let us consider now the Breuer-Hall map which is a generalisation the
reduction map%
\begin{equation} 
B:M(n,%
\mathbb{C}
)\rightarrow M(n,%
\mathbb{C}
);\quad B(X)=tr(X)id_{n}-A-UX^{T}U^{+},\quad X\in M(n,%
\mathbb{C}
),.
\end{equation}%
\begin{equation} 
X\in M(n,%
\mathbb{C}
),\quad U\in U(n),\quad U^{T}=-U.
\end{equation}%
So in order to construct a Breuer-Hall map one have to construct an unitary and anti-symmetric matrix.
It appears that one can construct a large class of unitary anti-symmetric
matrices $U\in U(n)$ (in fact orthogonal), using permutations from $S(n).$
We have

\begin{proposition}
Let $n=2k$. We divide the set $\{1,...,n\}$ into two disjoint subsets $O$
and $P$, where the first one contains all odd numbers from $\{1,...,n\}$ and
the second one contains all even numbers from $\{1,...,n\}$. The permutation
$\sigma =(o_{1}p_{1}),...,(o_{n-1}p_{n-1})\in S(n)$, where $o_{i}\in O$ and $%
p_{i}\in P$ is involutive and the matrix
\begin{equation} 
U^{\sigma }=((-1)^{j}\delta _{\sigma (i)j})
\end{equation}%
is unitary (orthogonal) and anti-symmetric.
\end{proposition}

\begin{remark}
So we have a large class of such unitary and antisymmetric matrices, which
are however, orthogonally similar. Note that the permutations $\sigma
=(o_{1}p_{1}),...,(o_{n-1}p_{n-1})\in S(n)$, where $o_{i}\in O$ and $%
p_{i}\in P$ belongs to the regular representation (i.e. permutational) of
the group $(%
\mathbb{Z}
_{2})^{\times n}$, so it is an element of a group of CDPs. In the following we will
use a particular, more convenient form of such a permutations, which looks $%
\sigma =(1p_{1}),...,(n-1p_{n-1})$.
\end{remark}

Now we are to formulate the main result of this section, which may be checked by a direct calculation

\begin{theorem}
Let $\Sigma _{2^{n}},$ $n=2l$ is a regular representation of the group $(%
\mathbb{Z}
_{2})^{\times n}$, so it is CDP, whose elements are compositions of
disjoint transpositions only i.e. $\Sigma _{2^{n}}=\{\sigma
_{i}\}_{i=1}^{2^{n}}:\sigma _{i}=(i_{1}j_{1})...$.$(i_{l}j_{l}).$ We choose
the permutation $\sigma _{p_{1}}=(1p_{1}),...,(n-1p_{n-1})\in $ $\Sigma
_{2^{n}}$, where $p_{i}\in P.$ Next let $\mathcal{A=}\{A^{k},k=1,...,n\}$ be
such that%
\begin{equation} 
A^{1}=(a_{ij}^{1}):a_{ii}^{1}=a_{2k-1p_{2k-1}}^{1}=a_{p_{2k-12}2k-1}^{1}=0,%
\quad k=1,..,l,\,i=1,..,n;
\end{equation}%
\begin{equation} 
a_{ij}^{1}=-1,\quad i,j\neq 2k-1,p_{2k-1},
\end{equation}

\begin{equation} 
A^{p_{1}}=0.
\end{equation}%
and
\begin{equation} 
\forall k\neq 1,p_{1}\qquad A^{k}=(a_{ij}^{k})=%
\begin{array}{c}
a_{ij}^{k}= \left\{ \begin{array}{ll} (-1)^{\sigma _{p_{1}}(i)+j}\delta _{\sigma _{i}\sigma
_{p_{1}}(j)k}:i\neq j \\
a_{ii}^{k}=1 \end{array} \right.
\end{array}%
\end{equation}%
then we have
\begin{equation} 
\rho \lbrack \mathcal{A},\Sigma _{n}\mathcal{]=}\sum_{ij=1}^{n}e_{ij}\otimes
B(e_{ij}),
\end{equation}%
where
\begin{equation} 
B(A)=tr(A)id_{n}-A-U^{\sigma _{p_{1}}}A^{T}(U^{\sigma _{p_{1}}})^{+},\quad
A\in M(n,%
\mathbb{C}
),
\end{equation}%
and $U^{\sigma _{p_{1}}}$ is an orthogonal, anti-symmetric matrix defined in
the last Proposition.
\end{theorem}

Interestingly, for $n=4$ the operator corresponding to the Breuer-Hall map belongs to both classes

\begin{equation}
 \sum_{i,j=0}^3 e_{ij} \ot B(e_{ij}) =  \left( \begin{array}{cccc|cccc|cccc|cccc}
 \cdot& \cdot& \cdot& \cdot& \cdot& \cdot& \cdot& \cdot& \cdot& \cdot& 1& \cdot& \cdot& \cdot& \cdot& 1\\
 \cdot& \cdot& \cdot& \cdot& \cdot& \cdot& \cdot& \cdot& \cdot& \cdot& \cdot& \cdot& \cdot& \cdot& \cdot& \cdot\\
 \cdot& \cdot& 1& \cdot& \cdot& \cdot& \cdot& \cdot& \cdot& \cdot& \cdot& \cdot& \cdot& -1& \cdot& \cdot\\
 \cdot& \cdot& \cdot& 1& \cdot& \cdot& \cdot& \cdot& \cdot& 1& \cdot& \cdot& \cdot& \cdot& \cdot& \cdot  \\ \hline
 \cdot& \cdot& \cdot& \cdot& \cdot& \cdot& \cdot& \cdot& \cdot& \cdot& \cdot& \cdot& \cdot& \cdot& \cdot& \cdot \\
 \cdot& \cdot& \cdot& \cdot& \cdot& \cdot& \cdot& \cdot& \cdot& \cdot& 1& \cdot& \cdot& \cdot& \cdot& 1 \\
 \cdot& \cdot& \cdot& \cdot& \cdot& \cdot& 1& \cdot& \cdot& \cdot& \cdot& \cdot& 1& \cdot& \cdot& \cdot \\
 \cdot& \cdot& \cdot& \cdot& \cdot& \cdot& \cdot& 1& -1& \cdot& \cdot& \cdot& \cdot& \cdot& \cdot& \cdot  \\ \hline
 \cdot& \cdot& \cdot& \cdot& \cdot& \cdot& \cdot& -1& 1& \cdot& \cdot& \cdot& \cdot& \cdot& \cdot& \cdot \\
 \cdot& \cdot& \cdot& 1& \cdot& \cdot& \cdot& \cdot& \cdot& 1& \cdot& \cdot& \cdot& \cdot& \cdot& \cdot \\
 1& \cdot& \cdot& \cdot& \cdot& 1& \cdot& \cdot& \cdot& \cdot& \cdot& \cdot& \cdot& \cdot& \cdot& \cdot \\
 \cdot& \cdot& \cdot& \cdot& \cdot& \cdot& \cdot& \cdot& \cdot& \cdot& \cdot& \cdot& \cdot& \cdot& \cdot& \cdot \\ \hline
 \cdot& \cdot& \cdot& \cdot& \cdot& \cdot& 1& \cdot& \cdot& \cdot& \cdot& \cdot& 1& \cdot& \cdot& \cdot \\
 \cdot& \cdot& -1& \cdot& \cdot& \cdot& \cdot& \cdot& \cdot& \cdot& \cdot& \cdot& \cdot& 1& \cdot& \cdot \\
 \cdot& \cdot& \cdot& \cdot& \cdot& \cdot& \cdot& \cdot& \cdot& \cdot& \cdot& \cdot& \cdot& \cdot& \cdot& \cdot \\
 1& \cdot& \cdot& \cdot& \cdot& 1& \cdot& \cdot& \cdot& \cdot& \cdot& \cdot& \cdot& \cdot& \cdot& \cdot \end{array}
 \right)\ ,
\end{equation}

\section{Conclusions}

In this paper we provide a class of sets of Completely Different Permutations (CDPs) which define a substantial generalization of the circular group $C(n) = \{ {\rm id},c,c^2,\ldots,c^{n-1} \}$, where $c=(0,1,\ldots,n-1)$. A class of CDPs enjoys several interesting properties analysed in Section II. This class is used to construct a bipartite operators acting on $\mathcal{H} \ot \mathcal{H}$, with $\mathcal{H}$ being an $n$ dimensional Hilbert space. The crucial observation shows that if $A$ is a bipartite operator corresponding to some abelian group of CDPs, then its partial transposition $(\oper \ot T)A$ corresponds to another abelian group of CDPs. Therefore, it may be used to construct and classify some classes of PPT states.  Interestingly several well known linear maps (reduction or Breuer-Hall maps) are related to sets of CDPs as well via Choi-Jamio{\l}lkowski isomorphism.


\section*{Acknowledgements} DC and GS were partially supported by the National Science Centre project 2015/19/B/ST1/03095. MM acknowledge support of research project "Quantum phenomena: Between the whole and the parts" of the John Templeton Foundation.


\begin{thebibliography}{99}

\bibitem{QIT} M. A. Nielsen and I. L. Chuang, Quantum Computation and
Quantum Information (Cambridge University Press, Cambridge,
2000).

\bibitem{HHHH} R. Horodecki, P. Horodecki, M. Horodecki,  and K. Horodecki, Rev.
Mod. Phys. {\bf 81}, 865 (2009).

\bibitem{NP}  L. Gurvits. Classical deterministic complexity of edmonds problem and quantum entanglement.
In Proc. of the 35th ACM symp. on Theory of comp., pp. 10-19, New York, 2003. ACM Press.

\bibitem{Guhne}  O. G\"uhne  and G. T\'oth,  Phys. Rep. {\bf 474}, 1 (2009).

\bibitem{TOPICAL} D. Chru\'sci\'nski and G. Sarbicki, J. Phys. A: Math. Theor. {\bf 47}, 483001 (2014).

\bibitem{circ} D. Chru\'sci\'nski and A. Kossakowski, Phys. Rev. A {\bf 76}, 032308 (2007).

\bibitem{ISO} M. Horodecki and P. Horodecki, Phys. Rev. A {\bf 59}, 4206 (1999).

\bibitem{Werner} R. F. Werner, Phys. Rev. A {\bf 40}, 4277 (1989).

\bibitem{CJ-1} M.-D. Choi,Linear Algebr. Appl. {\bf 12},  95 (1975).

\bibitem{CJ-2} A. Jamio{\l}kowski, Rep. Math. Phys. {\bf 3},  275 (1972).

\bibitem{Pillis} J. de Pillis,  Pac. J. Math. {\bf 23}, 129 (1967).

\bibitem{B} H.-P. Breuer, Phys. Rev. Lett. {\bf 97},  0805001 (2006).

\bibitem{H} W. Hall W, J. Phys. A: Math. Gen. {\bf 39}, 14119(2006).

\bibitem{Paulsen} V. Paulsen, {\it Completely Bounded Maps and Operator
Algebras}, Cambridge University Press, 2003.

\bibitem{KYE}  S.-H.  Kye, Rev. Math. Phys. {\bf 25}, 1330002 (2013).



\bibitem{CMP} D. Chru{\'s}ci{\'n}ski and A. Kossakowski, Comm. Math. Phys. {\bf 290}, 1051 (2009).


\bibitem{MM} M. Mozrzymas, M. Studzi\'nski, and N.  Datta, J. Math. Phys. {\bf 58}, 052204 (2017).

\bibitem{Kasia} K. Siudzi\'nska and D. Chru\'sci\'nski, {\em Quantum channels irreducibly covariant with respect to the finite group generated by the Weyl operators}, arxiv:1711.10823.

\bibitem{RR} K. Chen and L. A. Wu, Quantum Inf. Comput. {\bf 3}, 193 (2003).


\end{thebibliography}
\end{document}